\documentclass[envcount]{llncs}
\usepackage{amsmath}
\usepackage{amssymb}
\usepackage{paralist}
\usepackage{upref}
\usepackage{appendix}

\newtheorem{ob}{Observation}

\usepackage{graphicx}

\title{Improved Complexity Results on $k$-Coloring $P_t$-Free Graphs
\thanks{An extended abstract of this paper has appeared in the proceedings of MFCS 2013.}}
\author{Shenwei Huang}

\institute{School of Computing Science\\
Simon Fraser University, Burnaby B.C., V5A 1S6, Canada\\
\email{shenweih@sfu.ca}}

\date{}

\begin{document}

\maketitle

\begin{abstract}
A graph is $H$-free if it does not contain an induced subgraph isomorphic
to $H$. We denote by $P_k$ and $C_k$ the path and the cycle on $k$ vertices, respectively.
In this paper, we prove that
4-COLORING is NP-complete for $P_7$-free graphs, and that 5-COLORING is NP-complete
for $P_6$-free graphs. These two results improve all previous results on $k$-coloring
$P_t$-free graphs, and almost complete the classification of complexity of $k$-COLORING $P_t$-free graphs
for $k\ge 4$ and $t\ge 1$, leaving as the only missing case 4-COLORING $P_6$-free graphs.
We expect that 4-COLORING is polynomial time solvable for $P_6$-free graphs;
in support of this, we describe a polynomial time algorithm for 4-COLORING $P_6$-free graphs
which are also $P$-free, where $P$ is the graph obtained from $C_4$ by adding a new vertex
and making it adjacent to exactly one vertex on the $C_4$.
\end{abstract}

\section{Introduction}

We consider computational complexity issues related to vertex
coloring problems restricted to $P_k$-free graphs. It is well known that
the usual $k$-COLORING problem is NP-complete for any fixed $k\ge 3$. Therefore,
there has been considerable interest in studying its complexity when restricted
to certain graph classes. One of the most remarkable results in this respect
is that the chromatic number of perfect graphs can be decided in polynomial time. More information
on this classical result and related work on coloring problems restricted to graph
classes can be found in several surveys, e.g, \cite{Randerath,Tuza}.

We continue the study of $k$-COLORING problem for
$P_t$-free graphs. This problem has been given wide attention in recent years
and much progress has been made through substantial efforts by different groups
of researchers \cite{Fomin,2P_3,Broersma,Dabrowski,Short Cycle,Hoang,Lozin,Kral,Randerath 2007,Schiermeyer,Woe}.
We summarize these results and explain our new results below.

We refer to \cite{Bondy} for standard graph theory terminology
and \cite{Garey} for terminology on computational complexity.
Let $G=(V,E)$ be a graph and $\mathcal{H}$ be a set of graphs.
We say that $G$ is {\em $\mathcal{H}$-free} if $G$ does not contain
any graph $H\in \mathcal{H}$ as an induced subgraph. In particular,
if $\mathcal{H}=\{H\}$ or $\mathcal{H}=\{H_1,H_2\}$ , we simply
say that $G$ is $H$-free or $(H_1,H_2)$-free. Given any positive integer $t$,
let $P_t$ and $C_t$ be the path and cycle on $t$ vertices, respectively.
A {\em linear forest} is a disjoint union of paths. We denote by $G+H$
the disjoint union of two graphs $G$ and $H$.
We denote the complement of $G$ by $\bar{G}$.
The neighborhood of a vertex $x$ in $G$ is denoted by $N_G(x)$, or simply $N(x)$
if the context is clear. Given a vertex subset $S\subseteq V$ we denote by $N_S(x)$
the neighborhood of $x$ in $S$, i.e., $N_S(x)=N(x)\cap S$. For two disjoint
vertex subsets $X$ and $Y$ we say that $X$ is {\em complete,} respectively {\em anti-complete},
to $Y$ if every vertex in $X$ is adjacent, respectively non-adjacent, to every vertex in $Y$.
The {\em girth} of a graph $G$ is the length of the shortest cycle.

A {\em $k$-coloring} of a graph $G=(V,E)$ is a mapping $\phi:V\rightarrow \{1,2,\ldots,k\}$
such that $\phi(u)\neq \phi(v)$ whenever $uv\in E$. The value $\phi(u)$ is usually referred to
as the {\em color} of $u$ under $\phi$. We say $G$ is {\em $k$-colorable} if $G$ has a $k$-coloring.
The problem $k$-COLORING asks if an input graph admits an $k$-coloring.
The LIST $k$-COLORING problem asks if an input graph $G$ with lists $L(v)\subseteq \{1,2,\ldots,k\}$,
$v\in V(G)$, has a coloring $\phi$ that {\em respects} the lists, i.e.,
$\phi(v)\in L(v)$ for each $v\in V(G)$.

In the {\em pre-coloring extension of $k$-coloring} we assume that
(a possible empty) subset $W\subseteq V$ of $G$ is pre-colored with
$\phi_W:W\rightarrow \{1,2,\ldots,k\}$ and the question is whether
we can extend $\phi_W$ to a $k$-coloring of $G$. We denote the problem
of pre-coloring extension of $k$-coloring by $k^*$-COLORING.
Note that $k$-COLORING is a special case of $k^*$-COLORING, which in turn
is a special case of LIST $k$-COLORING.

Kami\'{n}ski and Lozin \cite{Lozin} showed that, for any fixed $k\ge 3$,
the $k$-COLORING problem is NP-complete for the class of graphs of girth
at least $g$ for any fixed $g\ge 3$.
Their result has the following immediate consequence.

\begin{theorem}[\cite{Lozin}]\label{th1}
For any $k\ge 3$, the $k$-COLORING problem is NP-complete for the class of $H$-free graphs whenever $H$
contains a cycle.
\end{theorem}

Holyer \cite{Holyer} showed that 3-COLORING is NP-complete for line graphs.
Later, Leven and Galil \cite{Leven} extended this result by showing that
$k$-COLORING is also NP-complete for line graphs for $k\ge 4$. Because
line graphs are claw-free, these two results together have the following consequence.

\begin{theorem}[\cite{Holyer,Leven}]\label{th2}
For any $k\ge 3$, the $k$-COLORING problem is NP-complete for the class of $H$-free graphs whenever $H$
is a forest with a vertex of degree at least 3.
\end{theorem}

Due to Theorems \ref{th1} and \ref{th2}, only the case in which $H$ is a linear forest remains.
In this paper we focus on the case where $H$ is a path. The $k$-COLORING problem is trivial for $P_t$-free
graphs when $t\le 3$. The first non-trivial case is $P_4$-free graphs.
It is well known that $P_4$-free graphs (also called {\em cographs})
are perfect and therefore can be colored optimally in polynomial time by Gr\"{o}tschel et al. \cite{Lovasz}.
Alternatively, one can color cographs using the {\em cotree representation} of a cograph, see, e.g., \cite{Randerath}.
Ho\`{a}ng et al. \cite{Hoang} developed an elegant recursive algorithm showing that the
$k$-COLORING problem can be solved in polynomial time for $P_5$-free graphs for any
fixed $k$.

Woeginger and Sgall \cite{Woe} proved that 5-COLORING is NP-complete for $P_8$-free graphs
and 4-COLORING is NP-complete for $P_{12}$-free graphs.
Later, Le et al. \cite{Randerath 2007} proved that 4-COLORING is NP-complete for $P_{9}$-free graphs.
The sharpest results so far are due to Broersma et al. \cite{Fomin,Broersma}.
\begin{theorem}[\cite{Broersma}]\label{4-COL}
4-COLORING is NP-complete for $P_8$-free graphs and $4^*$-COLORING is NP-complete for $P_7$-free graphs.
\end{theorem}

\begin{theorem}[\cite{Fomin}]\label{5-COL}
6-COLORING is NP-complete for $P_7$-free graphs and $5^*$-COLORING is NP-complete for $P_6$-free graphs.
\end{theorem}

In this paper we strengthen these NP-completeness results.
We prove that 5-COLORING is NP-complete for $P_6$-free graphs
and that 4-COLORING is NP-complete for $P_7$-free graphs.
We shall develop a novel general framework of reduction
and prove both results simultaneously in Section 2.
This leaves the $k$-COLORING problem for $P_t$-free graphs
unsolved only for $k=4$ and $t=6$, except for 3-COLORING.
(The complexity status of 3-COLORING $P_t$-free graphs for $t\ge 7$ is open.
It is even unknown whether there exists a fixed integer $t\ge 7$
such that 3-COLORING $P_t$-free graphs is NP-complete.)
We will focus on the case $k=4$ and $t=6$.
In Section 3, we shall explain why
the framework established in Section 2 is not sufficient to prove the NP-completeness of 4-COLORING
for $P_6$-free graphs. However, we were able to develop a polynomial time algorithm for
4-COLORING $(P_6, P)$-free graphs.
These two results suggest that 4-COLORING might be polynomially solvable for $P_6$-free graphs.
Finally, we give some related remarks in Section 4.

\section{The NP-completeness Results}


In this section, we shall prove the following main results.
\begin{theorem}\label{Mine 1}
5-COLORING is NP-complete for $P_6$-free graphs.
\end{theorem}

\begin{theorem}\label{Mine 2}
4-COLORING is NP-complete for $P_7$-free graphs.
\end{theorem}

Instead of giving two independent proofs for Theorems \ref{Mine 1} and \ref{Mine 2},
we provide a unified framework.
The {\em chromatic number} of a graph $G$,
denoted by $\chi(G)$, is the minimum positive integer $k$ such that $G$ is $k$-colorable.
The {\em clique number} of a graph $G$,
denoted by $\omega(G)$, is the maximum size of a clique in $G$.
A graph $G$ is called {\em k-critical} if $\chi(G)=k$ and $\chi(G-v)<k$
for any vertex $v$ in $G$. We call a $k$-critical graph {\em nice}
if $G$ contains three independent vertices $\{c_1,c_2,c_3\}$ such that
$\omega(G-\{c_1,c_2,c_3\})=\omega(G)=k-1$.
For instance, any odd cycle of length at least 7
with any its three independent vertices is a nice 3-critical graph.

Let $I$ be any 3-SAT instance with variables $X=\{x_1,x_2,\ldots,x_n\}$
and clauses $\mathcal{C}=\{C_1,C_2,\ldots,C_m\}$,
and let $H$ be a nice $k$-critical graph with three specified independent vertices $\{c_1,c_2,c_3\}$.
We construct the graph $G_{H,I}$ as follows.

$\bullet$ Introduce for each variable $x_i$ a {\em variable component} $T_i$ which is isomorphic to
$K_2$, labeled by $x_i\bar{x_i}$. Call these vertices {\em $X$-type}.

$\bullet$ Introduce for each variable $x_i$ a vertex $d_i$. Call these vertices {\em $D$-type}.

$\bullet$ Introduce for each clause $C_j=y_{i_1}\vee y_{i_2}\vee y_{i_3}$ a {\em clause component}
$H_j$ which is isomorphic to $H$, where $y_{i_t}$ is either $x_{i_t}$ or $\bar{x_{i_t}}$.
Denoted three specified independent vertices in $H_j$
by $c_{i_tj}$ for $t=1,2,3$. Call $c_{i_tj}$ {\em $C$-type} and all remaining vertices {\em $U$-type}.

For any $C$-type vertex $c_{ij}$ we call $x_i$ or $\bar{x_i}$ its {\em corresponding literal vertex},
depending on whether $x_i\in C_j$ or $\bar{x_i}\in C_j$.

$\bullet$ Make each $U$-type vertex adjacent to each $D$-type and $X$-type vertices.

$\bullet$ Make each $C$-type vertex $c_{ij}$ adjacent to $d_i$ and its corresponding literal vertex.

\begin{lemma}\label{iff}
Let $H$ be a nice $k$-critical graph.
Suppose $G_{H,I}$ is the graph constructed from $H$ and a 3-SAT instance $I$.
Then $I$ is satisfiable if and only if $G_{H,I}$ is $(k+1)$-colorable.
\end{lemma}
\begin{proof}
We first assume that $I$ is satisfiable and let $\sigma$ be a truth assignment
satisfying each clause $C_j$. Then we define a mapping $\phi:V(G)\rightarrow \{1,2,\ldots,k+1\}$ as follows.

$\bullet$ Let $\phi(d_i):=k+1$ for each $i$.

$\bullet$ If $\sigma(x_i)$ is TRUE, then $\phi(x_i):=k+1$ and  $\phi(\bar{x_i}):=k$.
Otherwise, let $\phi(x_i)=:k$ and  $\phi(\bar{x_i})=:k+1$.

$\bullet$ Let $C_j=y_{i_1}\vee y_{i_2}\vee y_{i_3}$ be any clause in $I$.
Since $\sigma$ satisfies $C_j$, at least one literal in $C_j$, say $y_{i_t}$ ($t\in \{1,2,3\}$), is TRUE.
Then the corresponding literal vertex of $c_{i_tj}$ receives the same color as $d_{i_t}$.
Therefore, we are allowed to color $c_{i_tj}$ with color $k$.
In other words, we let $\phi(c_{i_tj}):=k$.

$\bullet$ Since $H_j=H$ is $k$-critical, $H_j-c_{i_tj}$ has a $(k-1)$-coloring
$\phi_j:V(H_j-c_{i_tj})\rightarrow \{1,2,\ldots,k-1\}$. Let $\phi=:\phi_j$ on $H_j-c_{i_tj}$.

\noindent It is easy to check that $\phi$ is indeed a $(k+1)$-coloring of $G_{H,I}$.

Conversely, suppose $\phi$ is a $(k+1)$-coloring of $G_{H,I}$.
Since $H_1=H$ is a nice $k$-critical graph, the largest clique of $U$-type vertices in $H_1$ has size $k-1$.
Let $R_1$ be such a clique. Note that $\omega(G_{H,I})=k+1$ and $R=R_1\cup T_1$ is a clique of size $k+1$.
Therefore, any two vertices in $R$ receive different colors in any $(k+1)$-coloring of $G_{H,I}$. Without loss
of generality, we may assume $\{\phi(x_1),\phi(\bar{x_1})\}=\{k,k+1\}$.
Because every $U$-type vertex is adjacent to every $X$-type and $D$-type vertex,
we have the following three properties of $\phi$.

(P1) $\{\phi(x_i),\phi(\bar{x_i})\}=\{k,k+1\}$ for each $i$.

(P2) $\phi(d_i)\in \{k,k+1\}$ for each $i$.

(P3) $\phi(u)\in \{1,2,\ldots,k-1\}$ for each $U$-type vertex.

\noindent Next we construct a truth assignment $\sigma$ as follows.

$\bullet$ Set $\sigma(x_i)$ to be TRUE if $\phi(x_i)=\phi(d_i)$ and FALSE otherwise.

\noindent It follows from (P1) and (P2) that $\sigma$ is a truth assignment.
Suppose $\sigma$ does not satisfy $C_j=y_{i_1}\vee y_{i_2}\vee y_{i_3}$.
Equivalently, $\sigma(y_{i_t})$ is FALSE for each $t=1,2,3$. It follows from our definition
of $\sigma$ that the corresponding literal vertex of $c_{i_tj}$ receives
different color from the color of $d_{i_t}$ under $\phi$.
Hence, $\phi(c_{i_tj})\notin \{k,k+1\}$ for $t=1,2,3$
and this implies that $\phi$ is a $(k-1)$-coloring of $H_j=H$ by (P3).
This contradicts the fact that $\chi(H)=k$.  \qed
\end{proof}

\begin{lemma}\label{P_t-free}
Let $H$ be a nice $k$-critical graph.
Suppose $G_{H,I}$ is the graph constructed from $H$ and a 3-SAT instance $I$.
If $H$ is $P_t$-free where $t\ge 6$, then $G_{H,I}$ is $P_t$-free as well.
\end{lemma}
\begin{proof}
Suppose $P=P_t$ is an induced path with $t\ge 6$ in $G_{H,I}$.
We first prove the following claim.

\noindent {\bf Claim A.} {\em $P$ contains no $U$-type vertex.}

\noindent {\em Proof of Claim A.} Suppose that $u$ is a $U$-type vertex on $P$ that
lies in some clause component $H_j$. For any vertex $x$ on $P$ we denote by $x^-$ and $x^+$ the left
and right neighbor of $x$ on $P$, respectively. Let us first consider the case when $u$
is the left endvertex of $P$. If $u^+$ belongs to $H_j$, then $P\subseteq H_j$,
since $u$ is adjacent to all $X$-type and $D$-type vertices and $P$ is induced.
This contradicts the fact that $H$ is $P_t$-free.  Hence, $u^+$ is either $X$-type or $D$-type.
Note that $u^{++}$ must be $C$-type or $U$-type.
In the former case we conclude that $u^{+++}$ is $U$-type since $C$-type vertices are independent.
Hence, $|P|\le 3$ and this is a contradiction. In the latter case we have $|P|\le 4$ for the same
reason. Note that $|P|=4$ only if $P$ follows the pattern $U(X\cup D)UC$, namely the first
vertex of $P$ is $U$-type, the second vertex of $P$ is $X$-type or $D$-type, and so on.
Next we consider the case that $u$ has two neighbors on $P$.

{\bf Case 1.} Both $u^-$ and $u^+$ belong to $H_j$. In this case $P\subseteq H_j$ and this contradicts
the fact that $H=H_j$ is $P_t$-free.

{\bf Case 2.} $u^-\in H_j$ but $u^+\notin H_j$. Then $u^+$ is either $X$-type or $D$-type.
Since each $U$-type vertex is adjacent to each $X$-type and $D$-type vertex,
$u^-$ is a $C$-type vertex and hence it is an endvertex of $P$.
Now $|P|\le 2+4-1=5$.

{\bf Case 3.} Neither $u^-$ nor $u^+$ belongs to $H_j$. Now both $u^-$ and $u^+$ are
$X$-type or $D$-type. Since each $U$-type vertex is adjacent to each $X$-type and $D$-type
vertex, $P\cap U=\{u\}$ and $P\cap (X\cup D)=\{u^-,u^+\}$. Hence, $|P|\le 5$.
($|P|=5$ only if $P$ follows the pattern $C(X\cup D)U(X\cup D)C$).
\qed

Let $C_i$ (respectively $\bar{C_i}$) be the set of $C$-type vertices that connect to $x_i$ (respectively $\bar{x_i}$).
Let $G_{i}=G[\{T_i\cup\{d_i\}\cup C_i\cup \bar{C_i}\}]$.
Note that $G-U$ is disjoint union of $G_{i}$, $i=1,2,\ldots,n$.
By Claim A, $P\subseteq G_{i}$ for some $i$. Let $P'$ be a sub-path of $P$
of order 6 . Since $C_i \cup \bar{C_i}$ is independent, $|P'\cap (C_i \cup \bar{C_i})|\le 3$.
Hence, $|P'\cap (C_i \cup \bar{C_i})|= 3$ and thus $\{d_i,x_i,\bar{x_i}\}\subseteq P'$.
This contradicts the fact that $P'$ is induced
since $d_i$ has three $C$-type neighbors on $P'$. \qed
\end{proof}

Due to Lemmas \ref{iff} and \ref{P_t-free}, the following theorem follows.
\begin{theorem}\label{main th}
Let $t\ge 6$ be an fixed integer. Then $k$-COLORING is NP-complete for $P_t$-free graphs
whenever there exists a $P_t$-free nice $(k-1)$-critical graph. \qed
\end{theorem}

\noindent {\em Proof of Theorems \ref{Mine 1} and \ref{Mine 2}.}
Let $H_1$ be the graph shown in Figure \ref{H1} and let $H_2=C_7$ be the 7-cycle.
It is easy to check that $H_1$ is a $P_6$-free nice 4-critical graph
and that $H_2$ is a $P_7$-free nice 3-critical graph.
Applying Theorem \ref{main th} with $H=H_i$ ($i=1,2$) will complete our proof.
\qed

\begin{figure}[htbp]
\centering
\includegraphics[width=0.4\textwidth]{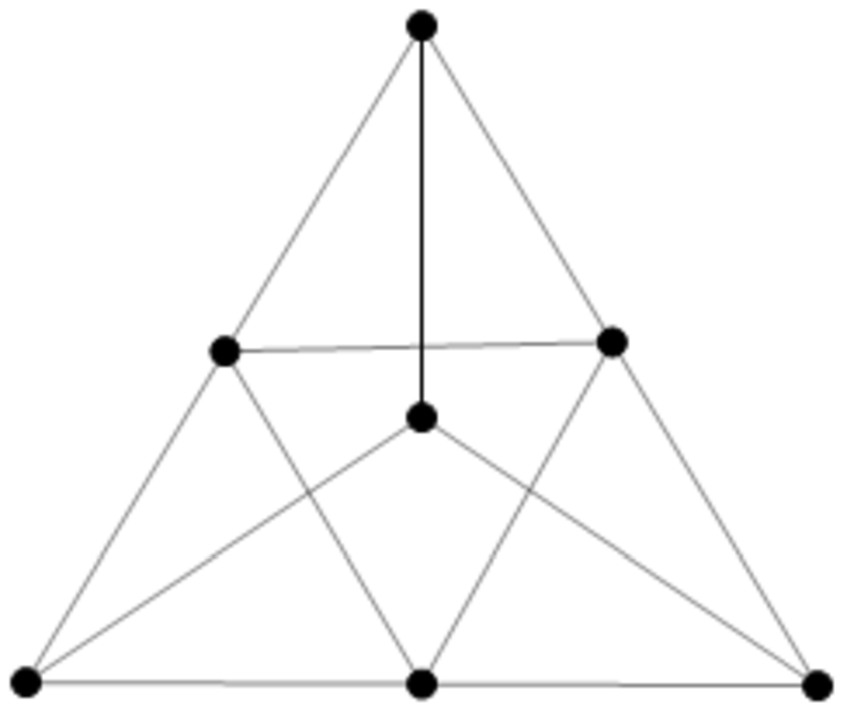}
\caption{$H_1$.}
\label{H1}
\end{figure}

\section{The Polynomial Result}

Having proved Theorems \ref{Mine 1} and \ref{Mine 2} the next question now is
whether 4-COLORING is NP-complete for $P_6$-free graphs. We show that the framework
established in Section 2 is not sufficient to prove the NP-completeness.
\begin{ob}
There is no $P_6$-free nice 3-critical graph.
\end{ob} 
\noindent {\bf Proof.} 
Suppose $H$ is $P_6$-free nice 3-critical graphs with $\{c_1,c_2,c_3\}$ being independent.
Since the only 3-critical graphs are odd cycles and $H$ is $P_6$-free, $H$ must be $C_5$.
But this contradicts the fact that $C_5$ contains at most two independent vertices. \qed

This negative result suggests that
4-COLORING $P_6$-free graphs might be solved in polynomial time.
But it seems difficult to prove this since
the usual techniques for 3-COLORING (see, e.g., \cite{Fomin,Broersma,Schiermeyer,Woe}) do not apply.
It turns out that the problem becomes easier if we forbid one more induced subgraph.
For example, if we consider $P_6$-free graphs that are also triangle-free, then the 4-COLORING problem
becomes trivial since every triangle-free $P_6$-free graph is 4-colorable (see, e.g., \cite{Randerath}).
Due to a result of Olariu \cite{Olariu} every component of a paw-free graph is either triangle-free
or a complete multipartite graph. Therefore, the result on 4-COLORING $(P_6,K_3)$-free graphs extends
to $(P_6,paw)$-free graphs. Golovach et al. \cite{Short Cycle} showed that $k$-COLORING is polynomial time
solvable for $(K_{r,s},P_t)$-free graphs for every fixed positive integer $k,r,s,t$.
In particular, 4-COLORING is polynomial time solvable for $(P_6,C_4)$-free graphs.

Next we shall prove 4-COLORING is polynomially solvable for $(P_6,P)$-free graphs
where $P$ is the graph obtained from $C_4$ by adding a new vertex and making it adjacent to exactly
one vertex on the $C_4$.
The coloring algorithm given below makes use of the following well-known lemma.
\begin{lemma}[\cite{Edwards}]\label{2-SAT}
Let $G$ be a graph in which every vertex has a list of colors of size at most 2.
Then checking whether $G$ has a coloring respecting these lists is solvable in polynomial time.
\end{lemma}

Randerath and Schiermeyer \cite{Schiermeyer} proved that one can decide in polynomial time that if a $P_6$-free
graph is 3-colorable.  Very recently Broersma et al. \cite{Fomin}
proved that the same applies to list coloring.
\begin{lemma}[\cite{Fomin}]\label{3-list col P_6}
LIST 3-COLORING can be solved in polynomial time for $P_6$-free graphs.
\end{lemma}

Now we are ready to prove our main result in this section.
\begin{theorem}\label{poly P_6}
4-COLORING is polynomially solvable for $(P_6,P)$-free graphs.
\end{theorem}
\begin{proof}
Let $G$ be a $(P_6,P)$-free graph.
It is easy to see that  $G$ is $k$-colorable if and only if every block of $G$ is $k$-colorable.
In addition, all blocks of $G$ can be found in linear time using depth first search.
So we may assume that $G$ is 2-connected. Also, we assume that $G$ is $K_5$-free.
Otherwise $G$ is not 4-colorable; and it takes $O(n^5)$ time to detect such a $K_5$.
We proceed by appealing to the polynomial time algorithm in \cite{Chundnovsky}
to distinguish two cases, according to whether or not $G$ is perfect.
If $G$ is perfect, one can even optimally color $G$ in polynomial time \cite{Lovasz}.
So, we may assume that $G$ is not perfect. By the Strong Perfect Graph Theorem \cite{SPGT}
and the fact that $G$ is $(P_6,P)$-free, $G$ must contain an induced $C_5$ or the complement of $C_7$.
We suppose first that $G$ is $C_5$-free. Now let $C$ be a complement of
an induced $C_7=v_0v_1\ldots v_6$.

We call a vertex $v\in V\setminus C$ an {\em $i$-vertex} if $|N(v)\cap C|=i$.
Let $S_i$ be the set of $i$-vertices where $0\le i\le 7$. Note that $G=V(C)\cup \bigcup_{i=0}^{7} S_i$.
For any $X\subseteq C$ we define $S(X)$ to be the vertices outside $C$ that has $X$ as their neighborhood on $C$.
Note that each $v_i$ is contained in an induced $C_4$ with other vertices on $C$.
Consequently, $S_1=\emptyset$ since $G$ is $P$-free. Further, $S_2(v_i,v_{i+1})=\emptyset$ or $v_iv_{i+2}v_{i-1}v_{i+1}t$ would induce a $C_5$ where $t$ is a vertex in $S_2(v_i,v_{i+1})$.
Next we shall prove that $S_0$ is anti-complete to $S_i$ for any $i\ge 2$. 
It is easy to check that
$S_0$ is anti-complete to $S_2(v_i,v_{i+2})$ and $S_2(v_i,v_{i+3})$ since $G$ is $P$-free.
Let $y\in S_0$ and assume that $y$ is adjacent to some vertex $x\in S_p$ for $p\ge 3$.
Suppose that $x\in S_3$. If $N_C(x)$ does not induce a triangle then there exist some index $i$
such that $v_{i-3}$ and $v_{i+3}$ belong to $N(x)$. Thus one of 
$\{v_{i-1},v_i,v_{i+1}\}$ is not adjacent to $x$, say $v_i$. Hence, $xv_{i-3}v_iv_{i+3}$ and the edge $yx$
induce a $P$ in $G$. So $N_C(x)$ is a triangle. Without loss of generality,
we may assume that $N(x)=\{v_{i-2},v_i,v_{i+2}\}$ for some $i$. 
Then $yxv_{i-2}v_{i-4}v_{i-1}v_{i-3}$ would induce a $P_6$.
Now suppose that $x$ is a $p$-vertex where $p\ge 4$. Note that $S_7=\emptyset$ or $G$
is not 4-colorable. If $x\in S_6$ then we may assume that $N_C(x)=C\setminus \{v_i\}$
and so $xv_{i-3}v_iv_{i+3}$ and $yx$ would induce a $P$. Suppose $x\in S_4$.
As $\omega(C)=3$ there are two vertices in $N_C(x)$ that are nonadjacent.
We may assume that they are $v_{i-3}$ and $v_{i+3}$. As $|N_C(x)|=4$ one of
$\{v_{i-1},v_i,v_{i+1}\}$ is not adjacent to $x$, say $v_i$. Hence, $xv_{i-3}v_iv_{i+3}$
and $yx$ induce a $P$ in $G$. Finally, $x\in S_5$ and we may assume that
$N_C(x)=\{v_{i-1},v_i,v_{i+1},v_{i-3},v_{i+3}\}$ otherwise we apply the argument in the case of $x\in S_4$.
But now $xv_{i-1}v_{i+2}v_i$ and $yx$ would induce a $P$.
Therefore, $S_0$ is anti-complete to $S_i$ for any $i\ge 2$. 
Hence, $S_0=\emptyset$ by the connectivity of $G$.
Now we can try all (at most $4^7$) possible 4-coloring of $C$ and for each such coloring we can
reduce the problem to 2-SAT by Lemma \ref{2-SAT}.

Now we assume that $C=v_0v_1v_2v_3v_4v_0$ is an induced $C_5$ in $G$.
We define $i$-vertex and $S_i$ as above.
By the $P$-freeness of $G$ we have the following facts.

$\bullet$ $S_5$ must be an independent set or $G$ is not 4-colorable.

$\bullet$ If $v\in S_3$, then $N_C(v)$ induces a $P_3$ and if $v\in S_2$, then $N_C(v)$ induces a $P_2$.

In the following all indices are modulo 5.
Let $S_3(v_i)$ be the set of 3-vertices whose neighborhood on $C$ is $\{v_i,v_{i-1},v_{i+1}\}$
and $S_2(v_i)$ be the set of 2-vertices whose neighborhood on $C$ is $\{v_{i-2},v_{i+2}\}$.
We also define $S_1(v_i)$ and $S_4(v_i)$ to be the set of 1-vertices and 4-vertices that has $v_i$ as their unique neighbor and non-neighbors on $C$, respectively.
Clearly, $S_p=\bigcup_{i=0}^{4} S_p(v_i)$ for $1\le p\le 4$.
It follows from the $P$-freeness of $G$ that $S_3(v_i)$ and $S_4(v_i)$ are cliques for each $i$.
So $|S_3(v_i)|\le 2$ and $|S_4(v_i)|\le 2$ since $G$ is $K_5$-free.
Hence, $|C\cup S_4\cup S_3|\le 25$ and  there are at most $4^{25}$ different 4-colorings of $C\cup S_4\cup S_3$.
Clearly, $G$ is 4-colorable if and only if there exists at least one such coloring that can be extended
to $G$. Therefore, it suffices to explain how to decide if a given 4-coloring $\phi$ of $C\cup S_4\cup S_3$ can be
extended to a 4-coloring of $G$ in polynomial time. Equivalently, we want to decide in polynomial time
if $G$ admits a list 4-coloring with input lists as follows.
\[L(v)=\left\{
\begin{array}{ll}
\{1,2,3,4\} & \mbox{if $v\notin C\cup S_4\cup S_3$.} \\
\phi(v) & \mbox{otherwise}
\end{array}
\right.\]

We say that vertices with list size 1 have been {\em pre-colored}.
Now we {\em update} the graph as follows.
For any pre-colored vertex $v$ and any $x\in N(v)$
we remove color $\phi(v)$ from the list of $x$, i.e., let $L(x):=L(x)\setminus \{\phi(v)\}$.
It is easy to see that $|L(x)|=1$ for each $x\in S_5$ and
$|L(x)|\le 2$ for any $x\in S_2$ after updating the graph.
Next we consider 0-vertices.

\noindent {\bf Claim B.} {\em $S_0$ is anti-complete to $S_1\cup S_2\cup S_4$.
In addition, any two 0-vertices that lie in
the same component of $S_0$ have exactly same neighbors in $S_3$.}

\noindent {\em Proof of Claim B.} The first claim follows directly from the $(P_6,P)$-freeness of $G$.
To prove the second claim it suffices to show that $N_{S_3}(x)=N_{S_3}(y)$
holds for any edge $xy\in E$ in $S_0$. By contradiction assume that there exists an edge
$xy$ in $S_0$ such that $x$ has a neighbor $z$ in $S_3$ with $yz\notin E$. Without loss of generality,
we assume $z\in S_3(v_0)$. Then $yxzv_1v_2v_3$ would induce a $P_6$ in $G$. \qed

Let $A$ be an arbitrary component of $S_0$. Suppose that $A$ is anti-complete to $S_3$.
Since $G$ is 2-connected, $A$ has at least two neighbors in $S_5$. In particular, $|S_5|\ge 2$.
Let $x$ and $y$ be two 5-vertices. Then any vertex $t\in A$ is either complete or anti-complete
to $S_5$ or we may assume that $t$ is adjacent to $x$ but not to $y$ and hence $v_0xv_2y$ plus
$tx$ would induce a $P$ in $G$. Now let $A'\subseteq A$ be the set of vertices that are complete to
$S_5$. Note that $A'\neq \emptyset$. Suppose that $A'\neq A$ and so there exist $t\in A\setminus A'$.
By the connectivity of $A$ we may assume that $t$ has a neighbor $t'\in A'$ and hence $v_0xt'y$ plus
$tt'$ would induce a $P$. Therefore, $A$ is complete to $S_5$. Thus the color of vertex in
$S_5$ does not show up on each vertex of $A$
after updating the graph and we can employ Lemma \ref{3-list col P_6}.
So we may assume that $A$ has a neighbor in $S_3$, say $x$.
By Claim B, $\phi(x)$ does not appear in the lists of vertices
in $A$ at all. So, we can decide if $\phi$ can be extended to $A$ in polynomial time
by Lemma \ref{3-list col P_6}. Since $S_0$ has at most $n$ components, it takes polynomial
time to check if $\phi$ can be extended to $S_0$.

Now we consider 1-vertices. Our goal is to branch on a subset of vertices in either $S_1$
or $S_2$ in such a way that after branching the vertices in $S_1$ that are not pre-colored are
anti-complete to the vertices in $S_2$ that are not pre-colored.
We want to accomplish such branching with only polynomial cost.
If we do achieve that then we can decide in polynomial time if $\phi$ can be extended to
$S_1$ and $S_2$ (independently) by applying Lemmas \ref{3-list col P_6} and \ref{2-SAT},
respectively. Therefore, in the following we focus on branching procedure and
refer to applying Lemmas \ref{3-list col P_6} and \ref{2-SAT} to $S_1$ and $S_2$
by saying "we are done". We start with the properties of $S_1(v_i)$'s.

\noindent {\bf Claim C.}
{\em $S_1(v_i)$ is complete to $S_1(v_{i+2})$ and anti-complete to $S_1(v_{i+1})$.
Further, if $S_1(v_i)$ and $S_1(v_{i+2})$ are both non-empty, then
$|S_1(v_i)|\le 3$ and $|S_1(v_{i+2})|\le 3$.}

\noindent {\em Proof of Claim C.} Without loss of generality, it suffices to prove the claim
for $S_1(v_0)$. Let $x\in S_1(v_0)$, $y\in S_1(v_1)$, and $z\in S_1(v_2)$.
If $xz\notin E$, then $xv_0v_4v_3v_2y$ would be an induced $P_6$ in $G$.
If $xy\in E$, then $v_0xyv_1v$ plus $v_1v_2$ would be an induced $P$.
Thus the first claim follows. Now suppose $|S_1(v_0)|\ge 4$.
Then $S_1(v_0)$ contains two nonadjacent vertices
$x$ and $x'$ since $G$ is $K_5$-free. Now $xv_0x'zx$ plus the edge $v_4v_0$
would be an induced $P$. \qed

It follows from Claim C that we can pre-color all 1-vertices if at least four $S_1(v_i)$ are non-empty,
or exactly three $S_1(v_i)$ are non-empty and three corresponding $v_i$'s induce a $P_2+P_1$ in $G$,
or exactly two $S_1(v_i)$ are non-empty and two corresponding $v_i$'s are non-adjacent.
In all these cases we update the graph and we are done.
The remaining cases are:
(1) exactly one $S_1(v_i)\neq \emptyset$;
(2) exactly two $S_1(v_i)$ are non-empty and two corresponding $v_i$'s are adjacent;
(3) exactly three $S_1(v_i)$ are non-empty and three corresponding $v_i$'s induce a $P_3$.

\noindent {\bf Claim D.} {\em $S_1(v_i)$ is anti-complete to all $S_2(v_j)$ for $j\neq i$. In addition,
if both $S_1(v_i)$ and $S_1(v_{i+1})$ are non-empty, then $S_1(v_i)$ is also anti-complete to $S_2(v_i)$. }

\noindent {\em Proof of Claim D.} It suffices to prove the claim for $S_1(v_0)$.
Let $x\in S_1(v_0)$, $y\in S_2(v_1)$ and $z\in S_2(v_2)$.
If $xy\in E$, then $xv_0v_4y$ plus $v_0v_1$ would induce a $P$.
If $xz\in E$, then $v_1v_2v_3v_4zx$ would induce a $P_6$.
By symmetry, the first part of the claim follows.
Suppose now $S_1(v_0)$ and $S_1(v_{1})$ are both non-empty.
Let $x\in S_1(v_0)$, $y\in S_1(v_1)$ and $z\in S_2(v_0)$.
If $xz\in E$, then $yv_1v_0xzv_3$ would induce a $P_6$ in $G$. \qed

It follows from Claim D that in the case (2) or (3)
(we can pre-color two of three $S_1(v_i)$'s and update the graph)
the 1-vertices that are not pre-colored are anti-complete to 2-vertices.
Note also that in the case (2) the two non-empty $S_1(v_i)$'s are anti-complete to each other.
So we are done in these two cases.
Finally, we assume $S_1(v_0)\neq \emptyset$ and $S_1(v_i)=\emptyset$ for $i\neq 0$.
If $|S_2(v_0)|\le 2$, then we pre-color it, update the graph and we are done.
So assume that $|S_2(v_0)|\ge 3$. If there is no edge between $S_1(v_0)$ and $S_2(v_0)$,
then we are done by Claim D.

Hence, assume that there is at least one edge between $S_1(v_0)$ and $S_2(v_0)$.

\noindent {\bf Claim E.} {\em $S_2(v_0)$ is a star.}

\noindent {\em Proof of Claim E.} Let $y\in S_2(v_0)$ be a neighbor of some vertex
$x\in S_1(v_0)$. Suppose $y'\in S_2(v_0)$ is not adjacent to $y$. Consider $y'v_2yxv_0v_4$.
Since $G$ is $P_6$-free, we have $xy'\in E$ and thus $xyv_2y'$ plus $v_0x$ induces a $P$.
Therefore, $y$ is adjacent to any other vertex in $S_2(v_0)$. Thus $S_2(v_0)$ is a star
since $G$ is $K_5$-free. \qed

By Claim E we can pre-color $S_2(v_0)$ since there are exactly two such colorings.
Finally we update the graph and we are done. Therefore, in any case we can decide in
polynomial time if $\phi$ can be extended to $S_1$ and $S_2$, and so the proof is complete.
\qed
\end{proof}

\section{Concluding Remarks}
We have proved that 4-COLORING is NP-complete for $P_7$-free graphs, and that
5-COLORING is NP-complete for $P_6$-free graphs. These two results
improve Theorems \ref{4-COL} and \ref{5-COL} obtained by Broersma et al. \cite{Fomin,Broersma}.
We have used a reduction from 3-SAT and establish a general framework. The construction
and the proof are simpler than those in previous papers.
As pointed out in Section 3,
however, they do not apply to 4-COLORING $P_6$-free graphs. On the other hand,
Golovach et al. \cite{Golovach} completed the dichotomy classification for 4-COLORING $H$-free graphs
when $H$ has at most five vertices.
The classification states that 4-COLORING is polynomially solvable for $H$-free graphs when $H$ is a linear
forest and is NP-complete otherwise. Note that linear forests on at most five vertices are all induced
subgraph of $P_6$. Thus, all the polynomial cases from \cite{Golovach} are for subclasses of
$P_6$-free graphs. We conjecture that it can be decided in polynomial time
if a $P_6$-free graph is 4-colorable.
\begin{conjecture}\label{4-col P_6}
4-COLORING can be solved in polynomial time for $P_6$-free graphs.
\end{conjecture}

We have proved that Conjecture \ref{4-col P_6} is true for $(P_6,P)$-free graphs, a subclass of $P_6$-free graphs.
Our proof makes use of certain ideas of Le et al. \cite{Randerath 2007} who
proved that 4-COLORING is polynomially solvable for $(P_5,C_5)$-free graphs,
and it also suggests new techniques that may be useful.
Furthermore,  Theorem \ref{poly P_6} may be interesting in its own right.
It suggests a new research  direction, namely
classifying the complexity of $k$-COLORING $(P_t, C_l)$-free
graphs for every integer combination of $k$, $l$ and $t$.
Since $k$-COLORING is NP-complete for $P_t$-free graphs
for even small $k$ and $t$, say Theorems \ref{Mine 1} and \ref{Mine 2},
it would be nice to know whether or not
forbidding short induced cycles makes problem easier (see, e.g., \cite {Huang}).
A recent result of Golovach et al. \cite{Short Cycle} showed that
forbidding $C_4$ does make problem easier.
In contrast, Golovach et al. \cite{Short Cycle}
showed that 4-COLORING is NP-complete for $(P_{164},C_3)$-free graphs.
They also determined a lower bound $l(g)$ for any fixed $g\ge 3$ such that
every $P_{l(g)}$-free graph with girth at least $g$ is 3-colorable.
Note that the girth condition implies the absence of all induced cycles
of length from 3 to $g-1$ in the graph.
Therefore, the last result can be viewed as
an answer to a restricted version of the problem we have formulated.

{\bf Acknowledgement.}
The author is grateful to his supervisor Pavol Hell for
many useful conversations related to these results
and helpful suggestions for greatly improving the presentation of the paper.
The author is grateful for Daniel Paulusma for communicating results from
\cite{Short Cycle} and valuable comments.

\end{document}